\documentclass[aps,prd,12pt,showpacs,preprintnumbers,groupedaddress,floatfix,nofootinbib]{revtex4}

\usepackage{amsmath}
\usepackage{graphicx}
\usepackage{color}

\def\beq{\begin{eqnarray}}
\def\eeq{\end{eqnarray}}
\def\bea{\begin{eqnarray}}
\def\eea{\end{eqnarray}}
\newcommand{\R}{{\rm I\!R}}
\newcommand{\tr}{\operatorname{tr}}

\newtheorem{theorem}{Theorem}
\newtheorem{lemma}{Lemma}
\newenvironment{proof}[1][Proof]{\textbf{#1.} }{\ \rule{0.5em}{0.5em}}

\begin{document}

\preprint{DCP-12-05}

\title{Symmetric texture-zero mass matrices with eigenvalues quark mass}
\author{A. Criollo}
\email{arturoc@uaeh.edu.mx}
\affiliation{\'Area Acad\'emica de Matem\'aticas y F\'{\i}sica, Universidad Aut\'onoma del Estado
de Hidalgo, Carr. Pachuca-Tulancingo Km. 4.5, C.P. 42184, Pachuca, Hgo.}

\author{R. Noriega-Papaqui}
\email{rnoriega@uaeh.edu.mx}
\affiliation{\'Area Acad\'emica de Matem\'aticas y F\'{\i}sica, Universidad Aut\'onoma del Estado
de Hidalgo, Carr. Pachuca-Tulancingo Km. 4.5, C.P. 42184, Pachuca, Hgo.}
\affiliation{Dual C-P Institute of High Energy Physics}


\begin{abstract}
\noindent Working within the context of texture-zeros mechanism for fermionic mass
matrices, we provide necessary and sufficient conditions on the characteristic
polynomial coefficients such that it has real, simple and positive roots. We
translate these conditions in terms of invariants from congruent matrices. Then,
all symmetric texture-zero matrices are counted and classified. Next
we apply  the result from the first part to analyze the three, two and one
zero texture matrices in a systematic way. Finally we solve analytically
the $V_{ckm}$ mixing matrix for the four zero sets; we also analyze the
$V_{ckm}$ for a particular case of four zero, four zero-perturbed and three zero sets.

\end{abstract}

\pacs{12.5.Ff; 02.10.Ud}

\maketitle
\section{Introduction}
In the Standard Model (SM) with $SU(2) \times U(1)$ as the gauge group of
electroweak interactions, the masses of quarks and charged leptons are contained
in the Yukawa Sector. After Spontaneous Symmetry Breaking (SSB), the mass
matrix is defined as:
$$
M_f = \frac{v}{\sqrt{2}} Y_f, \qquad
(f = u, d, l),
$$
where $v$ is the vacuum expectation value of the Higgs field and $Y_f$ are
the $3 \times 3$ Yukawa matrices.
The physical masses of the particles are defined as the eigenvalues
of the mass matrix $M_f$. Within the
SM context the mass matrix is unknown, the only trail of the quarks
mass matrices is the $V_{ckm}$ matrix, which is built by the product of left matrices
that diagonalize the $u$ and $d$-quark mass matrix.

In 1977 Harald Fritzsch proposed a phenomenological study
\cite{Fritzsch:1977za}, the so called
texture-zeros mechanism \footnote{For excellent reviews see \cite{Fritzsch:1999ee},
and references there in.}, that consist of looking for the simplest pattern
of mass matrices, which can result in a self-consistent way
and it reproduce the $V_{ckm}$ parameters obtained experimentally.
From all possible texture-zero matrices (symmetric, non-symmetric and
triangular\cite{Haussling:1997ue} matrices) we restrict our study to symmetric textures.
Mathematically speaking, a symmetric mass matrix always guarantees that
the physical masses are real, however, the positivity condition for the eigenvalues
is not fulfilled by any symmetric matrix, moreover a positive
definite matrix has real and positive eigenvalues, but not necessarily they are different.
In the texture-zeros formalism it is possible to have negative
eigenvalues, in this case, these negative signs can be removed with a rotation, however
in our proposal, in this paper  we will not consider this extra rotation,  we take as
starting point  strictly that all eigenvalues must  be the quark masses, in
other words, each eigenvalue must be real, different and positive.
Going in this direction, we discuss what kind of symmetric texture-zeros are self consistent
considering by definition that, the eigenvalues of the mass matrix are the masses of
quarks or charged leptons, and they must be positive (and  different)
real numbers.

The organization of this paper goes as follows. In Sec. II, we show analytically
how the mass matrices
appears in the SM context. In Sec. III, we find necessary and sufficient
conditions on the characteristic polynomial coefficients such that its roots
are real, simple and positive quantities. These conditions are rewritten in
terms of the invariants of the congruent matrices, \textit{i.e.} trace, determinant
and trace of the power matrix. In Sec. IV, we develop a simple notation that
counts and classifies the texture-zero matrices, and we show that all symmetric
matrices of $3\times 3$ can be grouped into 1-zero, 2-zero and 3-zero texture, in order to complete
the counting  the matrix without zeros is included. In Sec. V, we apply
systematically  the results of Sec. III to all matrices of the Sec. IV, and we
show what kind of texture matrices have real, different and positive eigenvalues. Finally
in Sec. VI, we derive analytical expressions for all the $V_{ckm}$ elements arising from
the 4-zero sets, then by choosing a particular case of a four zero set, we compute
the $V_{ckm}$ matrix, next we perturb this case in order  to improve the expressions for the $V_{ckm}$
 elements, finally we took this case to the three zero sets.

\section{Preliminaries}
In the Yukawa sector of the SM, the mass terms for quarks and charged leptons
can be expressed as
\begin{equation}
\label{LagMass}
\bar{u}_L M_u u_R  + \bar{d}_L M_d d_R  + \bar{l}_L M_l l_R,
\end{equation}
where $u_{L(R)}$, $d_{L(R)}$ and $l_{L(R)}$ are the left(right)-handed quark and charged leptons fields for the  u-sector $(u,c,t)$, d-sector $(d,s,b)$ and charged leptons $(e,\mu,\tau)$ respectively. $M_u$, $M_d$ and $M_l$ are the mass matrices. Expressing the above equation in terms of the physical fields, one diagonalize the mass matrices by bi-unitary transformations
\begin{eqnarray}
\label{Ec3}
\bar{M}_u &=& U^{\dagger}_{uL} \, M_u \, U_{uR} = Diag \left[ m_u, m_c, m_t \right], \nonumber \\
\bar{M}_d &=& U^{\dagger}_{dL} \, M_d \, U_{dR} = Diag \left[ m_d, m_s, m_b \right],   \\
\bar{M}_l &=& U^{\dagger}_{lL} \, M_l \, U_{lR} = Diag \left[ m_e, m_{\mu}, m_{\tau} \right], \nonumber
\end{eqnarray}
where $U_{f_L}$ and $U_{f_R}$ ($f = u, d, l$) are in general complex unitary matrices. The quantities $m_u, m_d, \dots$ etc. denote the eigenvalues of the mass matrices, i.e. the physical quark masses and they must have real and nonnegative quantities.

Re-expressing Eq. (\ref{LagMass}) in terms of physical fermion fields $(f'_{L (R)})$ as
\begin{equation}
\label{Lag-mass}
\bar{u'}_L \bar{M}_u u'_R  + \bar{d'}_L \bar{M}_d d'_R  + \bar{l'}_L \bar{M}_l l'_R,
\end{equation}
where $\bar{f'}_L\,= \,\bar{f}_{L} \,U_{fL}$ and $f'_R \,= \,U_{fR}^{\dagger} \, f_{R}$, ($f'= u', d', l'$).

\noindent Eq.(\ref{Ec3}) implies that $\bar{M}_f$ and $M_f$, $(f= u,\, d,\,l)$ are
congruent matrices, the relation of congruence is an equivalence relation, which implies a
space partition into cosets. Any two elements that belong at the same coset have the following
invariants: determinant, trace, trace of the power matrix, characteristic polynomial and their
eigenvalues, on the other hand,
if $\bar{M}_f$ and $M_f$ are congruent matrices then: $\det\bar{M}_f = \det M_f, \ \tr\bar{M}_f = \tr M_f, \ \tr\bar{M}_f^n = \tr M_f^n$,
 where $n$ is a positive integer,  $\det (\bar{M}_f -\lambda I) = \det (M_f-\lambda I)$
\footnote{In this work, we will denote the product $(\tr A)(\tr A)$
as $\tr^2A$. In the general case $(\tr A)^n = \tr^nA$ for $n$ positive integer.}\cite{FriedbergLA}.

\noindent Considering  $M_f$ as a $3\times 3$ symmetric matrix with real coefficients
then $\bar{M}_f$ is built as a diagonal matrix where its elements
are the eigenvalues of $M_f$, these eigenvalues are found as the roots its
characteristic polynomial. In the following section we give
conditions on the coefficients of the characteristic polynomial from $M_f$,
\textit{i.e} on the $M_f$ elements, such that this polynomial has three real,
positive and simple roots.

\section{Main Theorem}
The physical quark masses are defined as the eigenvalues of the mass matrix, from mathematical point of view, to obtain the quark masses it is necessary compute the characteristic equation and its roots are the quark masses. In this section we present the conditions over the characteristic polynomial coefficients such that the polynomial characteristics roots are real, positive and different. We translate these conditions in terms of invariants of congruent matrices as Trace and Determinant of the mass matrix.

\begin{theorem}
\label{theorem_1}
The polynomial of degree 3, $p(\lambda)=\lambda^3+a_2\lambda^2+a_1\lambda +a_0$ has three
different, real and positive roots if and only if the following conditions over its coefficients $a_0$,  $a_1$ $a_2$ hold.

\begin{enumerate}
\item $a_0,a_2<0<a_1.$\label{condition11}
\item $3a_1<a_2^2$.\label{condition12}
\item If $\lambda_4=\dfrac{-a_2+\sqrt{a_2^2-3a_1}}{3}$ and
$\lambda_5=\dfrac{-a_2-\sqrt{a_2^2-3a_1}}{3}$, then $p(\lambda_4)<0$ and $p(\lambda_5)>0$.\label{condition13}
\end{enumerate}
\end{theorem}

\begin{proof} See Appendix.
\end{proof}

We observe that in the condition \ref{condition13}, $\lambda_4$ and $\lambda_5$ are the roots of the first derivative of $p(\lambda)$, and therefore the condition \ref{condition12} implies that $\lambda_4$ and $\lambda_5$ are real numbers, in others words, $p(\lambda)$ has two critical points, this fact join to the condition 1 implies that $0<\lambda_5<\lambda_4$.

The condition \ref{condition13} ($p(\lambda_4)<0$ and $p(\lambda_5)>0$), means that the maximum value is positive and the minimum value is negative, and therefore $p(\lambda)$ has three real and different roots. This condition can be replaced by
\begin{equation}
-2(a_2^2-3a_1)^{3/2} <2a_2^3-9a_1a_2+27a_0<2(a_2^2-3a_1)^{3/2},\label{condicion3}
\end{equation}
the first inequality is obtaining by solving $p(\lambda_5)>0$ and the second one is obtained by solving $p(\lambda_4)<0$.
The condition (\ref{condicion3}) can be rewriting as
\begin{equation}
|2a_2^3-9a_1a_2+27a_0|<2(a_2^2-3a_1)^{3/2}.\label{condicion3_1}
\end{equation}

\noindent It is convenient to rewrite the theorem \ref{theorem_1}  in terms of the invariants of congruent matrices. This create directly a link between the matrix
elements and its eigenvalues which facilitates subsequent computations and
applications. To implement this fact, first we write the coefficients of its
characteristic polynomial $p(\lambda)$ in terms of its trace $(\tr M)$, trace
of the square matrix $(\tr M^2)$ and its determinant $(\det M)$ in the following form:

\begin{equation}
\label{polcar}
p(\lambda)=\lambda^3-\tr M \lambda^2+\dfrac{1}{2}\left[ \tr^2 M-\tr M^2 \right]\lambda -\det M.
\end{equation}

Now we are ready to present the main theorem of this section

\begin{theorem}
\label{T1}
A real, symmetric matrix $M$ has real, positive and different eigenvalues if and only if the following three conditions hold.
\begin{enumerate}
\item \begin{enumerate}
        \item $\det M>0$, \label{condition1a}
        \item $\tr M>0$, \label{condition1b}
        \item $\tr M^2<\tr^2M$. \label{condition1c}
        \end{enumerate}\label{condition1}
\item $\tr^2M<3\tr M^2$. \label{condition2}
\item $\left| \tr M (5\tr^2M-9\tr M^2)-54\det M \right| < \sqrt{2}(3\tr M^2-\tr^2M)^{3/2}$. \label{condition3}
\end{enumerate}
\end{theorem}

The theorem \ref{T1} will be applied to texture-zero matrices.

\section{Texture-zero Formalism}

A texture-zero matrix is a $3\times 3$ matrix with zeros in some entries,
the way to count them is the following: a zero in the main diagonal
add as 1, while zero off main diagonal add as $1/2$. We need to sum all
zeros for both mass matrices u-quarks and d-quarks. For example, given $M_u$ and $M_d$ as
$$
M_u = \begin{pmatrix}
       * & 0 & * \\
       0 & * & * \\
       * & * & 0
      \end{pmatrix}, \qquad
M_d = \begin{pmatrix}
       * & 0 & * \\
       0 & * & 0 \\
       * & 0 & 0
      \end{pmatrix}.
$$
For $M_u$ we have one zero in the main diagonal, we add $(+1)$
and $2$ zeros off main diagonal that add $ 1(= 1/2 + 1/2)$, then
$M_u$ has a 2-zero texture structure. Considering now $M_d$ we have a 3-zero texture structure
$(1 + 2)$. Then, this set of matrices is said to have a 5-zero texture structure.

We say: a parallel structure
for $M_u$ and $M_d$ mass matrices means that if $M_u$
has zeros in some places then $M_d$ has zeros in the same
position than $M_u$. Non-parallel structure is when $M_u$ and
$M_d$ not have the same parallel structure.

\subsection{Notation}
\noindent We start writing a symmetric matrix $M$ in the form:
$$M=
\left(%
\begin{array}{ccc}
   E & D & F \\
   D & C & B\\
   F & B & A\\
\end{array}%
\right). $$

\noindent This matrix is well determined by specifying six capital letters $(A, B, C, D, E, F)$ and their
corresponding positions, then we introduce the following notation:
\begin{itemize}
\item $M(x)$ is a matrix with a zero in the capital letter $x$, $(x = A, B, C, D, E, F)$.
\item $M(x,y)$ is a matrix with zeros in the capital letters $x$ and $y$, ($x$, $y$ = $A$,
$B$, $C$, $D$, $E$, $F$;  $x \neq y$).
\item $M(x,y,z)$ is a matrix with zeros in the capital letters $x$, $y$ and
$z$, $(x,y,z = A, B, C, D, E, F; \, x \neq y \neq z)$.
\end{itemize}
For example, a matrix with a zero in the position $F$ is:
$$
M(F)=
\left(%
\begin{array}{ccc}
   E & D &0 \\
   D & C & B\\
   0 & B & A\\
\end{array}%
\right),
$$

\noindent a matrix with zeros in the positions $C$ and $D$ is,
$$
M(C,D)=
\left(%
\begin{array}{ccc}
   E & 0 &F \\
   0 & 0 & B\\
   F & B & A\\
\end{array}%
\right),
$$

\noindent finally a matrix with zeros in the positions $C$, $D$ and $F$ is,
$$
M(C,D,F)=
\left(%
\begin{array}{ccc}
   E & 0 &0 \\
   0 & 0 & B\\
   0 & B & A\\
\end{array}%
\right).
$$

\noindent Using this notation, we are able to list all possible textures.

\noindent \textbf{1-zero texture structure}. \\
We have 6 different matrices, which are: \\
$M(A), \ M(C), \ M(E), \ M(B), \ M(D), \ M(F).$

\noindent\textbf{2-zero texture structure}. \\
In this case, we have 15 possibilities, which are:

$M(A,E), \ M(A,C), \ M(C,E), $

$ M(A,B), \ M(A,D), \ M(A,F),$

$ M(B,C), \ M(C,D), \ M(C,F), $

$ M(B,E), \ M(D,E), \ M(E,F),$

$ M(B,F), \ M(B,D), \ M(D,F). $

\noindent\textbf{3-zero texture structure}. \\
For this case, there are 20 different matrices, which are:

$M(A,B,C), \ M(A,C,F), \ M(A,C,D), $

$M(A,B,E), \ M(A,E,F), \ M(A,D,E), $

$M(B,C,E), \ M(C,E,F), \ M(C,D,E), $

$M(A,B,F), \ M(A,B,D), \ M(A,D,F), $

$M(B,C,F), \ M(B,C,D), \ M(C,D,F), $

$M(B,E,F), \ M(B,D,E), \ M(D,E,F), $

$M(A,C,E), \ M(B,D,F). $

Now we are ready to analyze which kind of textures have three different and positive eigenvalues, applying in each case one of the theorems presented in previous sections.

\section{Combined Analysis}

\noindent The aim of this section is give to quark mass matrices the structure of zero textures and
find which of these structures have real, positive and different eigenvalues. In order to start the
analysis in a systematic way, we need to implement
another sub-classification, which depends on whether the matrix has or not zeros in the
main diagonal, doing this, first we analyze the 3-zeros textures, after this, we
study the 2-zero textures and finally the 1-zeros textures.

\subsection{3-zero analysis}

\noindent According the sub-classification given above, the 3-zero textures present the following cases:
 \begin{enumerate}
   \item Without zeros in the main diagonal there is one case $M(B,D,F)$.
   \item With one zero in the main diagonal exist $9$ cases: $ M(A,B,F)$, $M(A,B,D)$, $M(A,D,F)$, $M(B,C,F)$, $M(B,C,D)$, $M(C,D,F)$, $M(B,E,F)$, $M(B,D,E)$, $M(D,E,F)$.
   \item With two zeros in the main diagonal there are $9$ cases: $M(A,B,C)$, $M(A,C,F)$, $M(A,C,D)$, $M(A,B,E)$, $M(A,E,F)$, $M(A,D,E)$, $M(B,C,E)$, $M(C,E,F)$, $M(C,D,E)$.
   \item With three zeros in the main diagonal we have only $1$ case ($M(A,C,E)$).
 \end{enumerate}

 We obtain a total of $20$ different possibilities. We only present the analysis of the following three cases.

\begin{itemize}
  \item Applying the Theorem \ref{T1} (\ref{condition1b}) the trivial $M(A,C,E)$ case is ruled out\footnote{In this work, we are looking for textures with  positive and different eigenvalues, therefore,
textures with two equal eigenvalues or one of them negative, we say that, they are ruled out}
  \item Now, we analyze the Fritzsch 6-zero texture given by $M(C,E,F)$ \cite{Fritzsch:1977za}.
Applying again the Theorem \ref{T1} (\ref{condition1b}) we must have $\tr M(C,E,F) = A > 0$, from
the condition (\ref{condition1a}) $\det M(C,E,F) = -D^2 A < 0$ that is a contradiction, because
of that this 6-zero texture is ruled out.
  \item Next, we analyze the following texture $M(A,D,F)$.
The condition (\ref{condition1b}) of the Theorem \ref{T1}  we have that
$\tr M(A,D,F) = C + E > 0$ and from (\ref{condition1a})
$\det M(A,D,F) = -E B^2  > 0, \ \Leftrightarrow \ E < 0 \ \Rightarrow \ C > 0 \ \Rightarrow \ EC <0$.
The condition (\ref{condition1c}) of the Theorem \ref{T1} implies that
$0 < E^2 + C^2 + 2B^2 < E^2 + C^2 + 2EC \ \Rightarrow \ 0 < EC$ and we have a contradiction
and this texture is also ruled out.
\end{itemize}

We have analyzed the others $17$ cases and we found that the only case that
is not excluded is $M(B,D,F)$, obviously being $A, \ B$ and $C$ the eigenvalues ($A \neq C \neq  E > 0$).

\subsection{2-zero analysis}

\noindent These kind of textures have the following cases:

\begin{enumerate}
  \item Without zeros in the main diagonal there are $3$ cases: $M(B,F)$, $M(B,D)$, $M(D,F)$.
  \item With one zero in the main diagonal exist $9$ cases: $M(A,B)$, $M(A,D)$, $M(A,F)$, $M(B,C)$, $M(C,D)$, $M(C,F)$, $M(B,E)$, $M(D,E)$, $M(E,F)$
  \item With two zeros in the main diagonal there are $3$ cases: $M(A,E)$, $M(A,C)$, $M(C,E)$.
\end{enumerate}
We present the analysis of some cases more representative:

\begin{itemize}
  \item We start with the matrix $M(C,E)$. If we compute $\tr^2 M(C,E)$,
$\tr M(C,E)^2$ and we apply the condition (\ref{condition1c}) of the Theorem \ref{T1}, we obtain:
$$
2(D^2 + F^2 + B^2) + A^2 < A^2,
$$
that is a contradiction. We have found that $M(A,E)$ and $M(A,C)$ are ruled out too.
  \item The second example is the Fritzsch 4-zero texture given by $M(E,F)$ \cite{Fritzsch:2002ga}.
From the Theorem \ref{T1} follows that the condition (\ref{condition1a}) $\det M(E,F) = -A D^2  > 0$
implies $A < 0$, and of the condition (\ref{condition1b}) $\tr M(E,F) = C + A > 0$ we
have that $C > 0$ and then $A C < 0$. Now we compute $\tr^2 M(C,E), \ \tr M(C,E)^2$ and
using the condition (\ref{condition1c}) of the Theorem \ref{T1}, we obtain:
$$
2(D^2 + B^2) + C^2 + A^2 < C^2 + A^2 + 2AC,
$$
then $AC > 0$, that is a contradiction.
\end{itemize}

We have analyzed the eight cases $M(A,B)$, $M(A,D)$, $M(A,F)$, $M(B,C)$, $M(C,D)$, $M(C,F)$, $M(B,E)$, $M(D,E)$ and we found that are ruled out.

The cases that are in agreement with the condition (\ref{condition1}) of the
Theorem \ref{T1} are $M(B,F)$, $M(B,D)$ and $M(D,F)$, this means that, it exist a range of values of
$(B,F)$, $(B,D)$ and $(D,F)$ where these textures have real, positive and different eigenvalues.

\subsection{1-zero analysis}

Here we only have two cases,

\begin{enumerate}
  \item Without zeros in the main diagonal belong three different possibilities $M(B)$, $M(D)$ and $M(F)$.
  \item With one zero in the main diagonal also belong three different possibilities $M(A)$, $M(C)$ and $M(E)$.
\end{enumerate}

We only present the analysis of $M(A)$. The condition (\ref{condition1b}) produces
$E+C>~0$, the condition (\ref{condition1a}) implies that $2BDF-B^2E-F^2C>0$ and the
condition (\ref{condition1c}) gives $2(B^2+D^2+F^2)+E^2+C^2<E^2+C^2+2EC$,
the last three inequalities are equivalents with
\begin{eqnarray}
  E+C &>& 0, \label{eq1}\\
  2BDF &>& B^2E+F^2C, \label{eq2}\\
  0<B^2+D^2+F^2 &<& EC, \label{eq3}
\end{eqnarray}
from (\ref{eq1}) and (\ref{eq3}) we have that $E>0$ and $C>0$, therefore
\begin{eqnarray}
  -2BF\sqrt{EC} &<& B^2E+F^2C, \label{eq4}\\
  2BF\sqrt{EC} &<& B^2E+F^2C, \label{eq5}
\end{eqnarray}
now if $BF>0$, the inequalities (\ref{eq3},\, \ref{eq5},\, \ref{eq2}) produce
the following chain of inequalities
\begin{equation*}
    2BF\sqrt{B^2+D^2+F^2}<2BF\sqrt{EC}<B^2E+F^2C<2BDF,
\end{equation*}
and then
\begin{equation*}
    \sqrt{B^2+D^2+F^2}<D,
\end{equation*}
that is a contradiction. If $BF<0$ use (\ref{eq4}). We have analyzed
the other 2 cases $M(C)$, $M(E)$ and we found that are ruled out.

The cases that are in agreement with the condition (\ref{condition1}) of the
Theorem \ref{T1} are $M(B)$, $M(D)$ and $M(F)$.

Summing up this section, the zero texture mass matrices that they have real, positive and
different eigenvalues are:

$$
M(B, \, F), \, M(B, \, D), \, M(D, \, F), \, M(B), \, M(D) \, \text{and} \, M(F).
$$

Our results are in agreement with \cite{Branco:1999nb}, where the authors
using Weak Basic Transformations they have shown that any symmetric
texture with (1,1) zero entry has at least one negative eigenvalue.

\section{$V_{ckm}$ Properties}

Another important quantity that any quark mass matrices need to satisfied
it is reproduce the experimental values of the $V_{ckm}$ for this reason, in this
section we analyze the $V_{ckm}$ phenomenology, in the first
part and  considering a set of four zeros for mass matrices, we
note the presence of zeros in the $V_{ckm}$ that depend if we have
a parallel and non parallel structures in the quark mass matrices, in
the second part we choose a particular non parallel case and
compute the $V_{ckm}$ matrix. In order to fit this $V_{ckm}$ matrix with the
experimental $V_{ckm}$ matrix we introduce a perturbation analysis.
Finally  we present a set of three zeros where the
$V_{ckm}$ fits numerically.

\subsection{$V_{ckm}$ from 4-zero texture set}

In the previous sections it was shown that $M(B,F), \ M(B,D)$ and $M(D,F)$
are matrices with simple, real and different eigenvalues. When the mass matrix
of u-type quarks and the mass matrix of and d-type
quarks both have a parallel structure ({\it e.g.} $M_u = M_u(B_u,F_u)$
and $M_d = M_d(B_d,F_d)$), one direct implication is that the $V_{ckm}$ has the
same texture structure as the mass matrices ($V_{ckm}= V_{ckm}(B_{ckm},F_{ckm})$)
and we cannot reproduce the experimental values of the $V_{ckm}$ elements
because of that, all these three cases are ruled out.

Now, if the mass matrix of u-type quarks and the mass matrix of and d-type
quarks have not a parallel structure, all nine cases were analyzed and
always we find one zero element (off main diagonal) in the $V_{ckm}$ matrix.
We present the case where the best fit of the $V_{ckm}$ is found, this is because
we can obtain analytic expressions as well as
a lot of information about the mass matrices. For this, we choose the mass
matrix $M(D,F)$ texture for u-type quarks, and the matrix $M(B,F)$ texture for
d-type quarks. Then we have that

$$
M_u=
\left(%
\begin{array}{ccc}
   m_u & 0 & 0 \\
   0 & C_u & B \\
   0 & B & A_u \\
\end{array}%
\right), \qquad
 M_d=
\left(%
\begin{array}{ccc}
   E_d & D & 0 \\
   D & C_d & 0 \\
   0 & 0 & m_b \\
\end{array}%
\right).
$$
From the appendix (\ref{x}) and (\ref{y}), the above matrices take the form:
 $$ \ M_u=
\left(%
\begin{array}{ccc}
   m_u & 0 &0 \\
   0 & \mu_{ct} + \sqrt{\delta_{tc}^2 -B^2} & B\\
   0 & B & \mu_{ct} - \sqrt{\delta_{tc}^2 -B^2}\\
\end{array}%
\right), $$

$$
M_d=
\left(%
\begin{array}{ccc}
   \mu_{ds} + \sqrt{\delta_{sd}^2 -D^2} & D & 0 \\
   D & \mu_{ds} - \sqrt{\delta_{sd}^2 -D^2} & 0 \\
   0 & 0 & m_b\\
\end{array}%
\right),
$$
where $\mu_{qi\,qj} = \dfrac{m_{qi}+m_{qj}}{2}$ and $\delta_{qi\,qj} = \dfrac{m_{qi}-m_{qj}}{2}$ 
(with $m_{qi}>m_{qj}$, $i,\,j = 1,2,3$ and $q=u,\,d$). The quantities 
$\mu_{qi\,qj}$ and $\delta_{qi\,qj}$ have a interesting physical meaning, the 
first one is the average mass, and for the second one we can rewriting 
as $2\delta_{qi\,qj} +m_{qj}=m_{qi}$, then $2\delta_{qi\,qj}$  is the quantity 
that distinguishes the masses, i.e. the particles $m_{qi}$ and
$m_{qj}$ are different because their mass are different and the factor of
difference is $2\delta_{qi\,qj}$. The matrices that diagonalize the mass matrices are
\begin{equation}
\label{UU_UD}
 \ U_u=
\left(%
\begin{array}{ccc}
   1 & 0 & 0 \\
   0 & \cos\beta & \sin\beta\\
   0 & -\sin\beta & \cos\beta\\
\end{array}%
\right), \qquad
\ U_d=
\left(%
\begin{array}{ccc}
   \cos\alpha & \sin\alpha & 0 \\
   -\sin\alpha & \cos\alpha & 0\\
   0 & 0 & 1\\
\end{array}%
\right),
\end{equation}
where $\sin\alpha = \frac{D}{\sqrt{D^2 +(y_d - m_d)^2}}$ and
$\sin\beta = \frac{B}{\sqrt{B^2 +(y_u - m_c)^2}}$.

Now we computing the $V_{ckm}= U_u^T U_d$ matrix

\begin{equation}
\label{Vckm_4tex}
 \ V_{ckm}=
\left(%
\begin{array}{ccc}
   \cos\alpha & \sin\alpha & 0 \\
   -\cos\beta \sin\alpha & \cos\beta \cos\alpha & -\sin\beta \\
   -\sin\beta \sin\alpha & \sin\beta\cos\alpha & \cos\beta\\
\end{array}%
\right).
\end{equation}
Setting:
\begin{equation}
\label{seno_alfa}
 \sin\alpha = V_{us} = \lambda, \qquad
 \sin\beta = -V_{cb} =  -A \lambda^2,
\end{equation}
where $\lambda$ is the Wolfenstein parameter and $A$ is a real number of order one.

The $V_{ckm}$ matrix takes the form:
$$ \ V_{ckm}=
\left(%
\begin{array}{ccc}
   1 -\frac{\lambda^2}{2} & \lambda & 0 \\
   -\lambda & 1 -\frac{\lambda^2}{2} & A \lambda^2 \\
   A \lambda^3 & -A \lambda^2 & 1 \\
\end{array}%
\right) + O(\lambda^4).
$$
With this election of texture structure of the mass matrices of quarks, we can
reproduce (in Wolfenstein parametrization) eight $V_{ckm}$ parameters and
the $(V_{ckm})_{13}$ element is zero. Now, with this information we can know
explicitly each element of the mass matrices, from (\ref{seno_alfa}) we have that
\begin{eqnarray}
\label{sin_alpha}
  \sin\alpha &=& \frac{D}{\sqrt{D^2 +(y_d - m_d)^2}} = V_{us}, \\
  \label{sin_beta}
  \sin\beta &=& \frac{B}{\sqrt{B^2 +(y_u - m_c)^2}} = - V_{cb},
\end{eqnarray}
the solutions for $D$ and $B$ are:
\begin{eqnarray}
\label{D_0}
 D_0 &=& \pm 2 \delta_{sd} V_{us}\sqrt{1 -V_{us}^2} \, \approx \pm 2 \delta_{sd} V_{us}, \\
 B_0 &=& \pm 2 \delta_{tc} V_{cb}\sqrt{1 -V_{cb}^2} \, \approx \pm 2 \delta_{tc} V_{cb},
\end{eqnarray}
and the mass matrices are:
\begin{equation}
\label{MMM}
 M_u=
\begin{pmatrix}
   m_u & 0 &0 \\
   0 & m_c + 2 \delta_{tc} V_{cb}^2  & \pm 2 \delta_{tc} V_{cb} \\
   0 & \pm 2 \delta_{tc} V_{cb} & m_t - 2 \delta_{tc} V_{cb}^2 \\
\end{pmatrix}, \quad
M_d=
\begin{pmatrix}
   m_d + 2 \delta_{sd} V_{us}^2 & \pm 2 \delta_{sd} V_{us} & 0 \\
   \pm 2 \delta_{sd} V_{us} & m_s - 2 \delta_{sd} V_{us}^2 & 0 \\
   0 & 0 & m_b\\
\end{pmatrix}.
\end{equation}

Finally the mass matrices can be written as:
\begin{eqnarray}
\label{MU_parts}
M_u &=& \bar{M}_u + 2 \delta_{tc} \, V_{cb}^2 \, \Delta M_u \pm 2 \delta_{tc} \, V_{cb} \, \delta M_u, \\
\label{MD_parts}
M_d &=& \bar{M}_d + 2 \delta_{sd} \, V_{us}^2 \, \Delta M_d \pm 2 \delta_{sd} \, V_{us} \, \delta M_d,
\end{eqnarray}
where the matrices $\Delta M_u$, $\Delta M_d$, $\delta M_u$ and $\delta M_d$ are given by:

$$
\Delta M_u =
\begin{pmatrix}
  0 & 0 & 0 \\
  0 & 1 & 0 \\
  0 & 0 & -1 \\
\end{pmatrix}, \quad
\Delta M_d =
\begin{pmatrix}
   1 & 0 & 0 \\
   0 & -1 & 0 \\
   0 & 0 & 0 \\
\end{pmatrix}, \quad
\delta M_u =
\begin{pmatrix}
  0 & 0 & 0 \\
  0 & 0 & 1 \\
  0 & 1 & 0 \\
\end{pmatrix}, \quad
\delta M_d =
\begin{pmatrix}
   0 & 1 & 0 \\
   1 & 0 & 0 \\
   0 & 0 & 0 \\
\end{pmatrix}.
$$

We observe that the mass matrices have three contributions; the first one
($\bar{M}$) comes from a diagonal matrix, where its elements correspond to
mass quarks, the second contribution ($\Delta M$) is a correction of diagonal
entries and it is characterized by the square of $(V_{ckm})_{12}$ and
$(V_{ckm})_{13}$ elements respectively. The last contribution ($\delta M$)
is off-diagonal correction characterized by the $(V_{ckm})_{12}$ and
$(V_{ckm})_{13}$ elements. Note that: off diagonal contribution is bigger
than the diagonal ones.

From (\ref{Vckm_4tex}) we can see that, we get one zero in $(V_{ckm})_{13}$ element, the experimental value for this element is around $0.00351$, this
invite us to apply perturbation theory to 4-zero texture (especially in the
example presented above) in order to remove this zero and get a better
approximation for this $V_{ckm}$ element.

\subsubsection{Perturbative analysis of 4-zero texture set}

As we saw in previous section, when we consider a 4- zero texture set as structure of mass
matrices of  quarks, the presence of zeros in the $V_{ckm}$ matrix is unavoidable, the aim
of this part of the paper is use perturbation theory to remove these zeros and get small
quantities.

We consider that quark mass matrices can be divide in two parts:
\begin{equation}
 M_q = M_{q(2T)} + \epsilon N_q,
\end{equation}
where $M_{q(2T)} $ is a 2-zero texture, $N_q$ is known mass matrix and $\epsilon_q$
a small parameter (See appendix for more details). The new contributions to $V_{ckm}$
matrix comes from a antisymmetric matrix $X_q$.

In the example presented before, where $M_u = M_u(D_u, F_u)$ and $M_d = M_d(B_d, F_d)$
are the structures for the quark mass matrices,
one can reproduce eight  experimental values of $V_{ckm}$ elements and one get that the
$(V_{ckm})_{13}$ element is zero.  To remove this zero first we consider a perturbation
on $M_u = M_u(D_u, F_u)$ and keeping $M_d = M_d(B_d, F_d)$  unchanged, after that, we
will interchange the roles.

We consider that, $M'_u$ mass matrix  differs a small quantity
\footnote{$|\epsilon \, a_u| \sim |\epsilon \, b_u| \ll |C_u|, \, |A_u|, \, |B|, \, m_u $}
$\epsilon$ from $M_u$ in the positions $(1,2), \ (2,1), \ (1,3)$ and $(3,1)$.
$$
 M'_u =
\left(%
\begin{array}{ccc}
   m_u & \epsilon \, a_u & \epsilon \, b_u \\
   \epsilon \, a_u & C_u & B \\
   \epsilon \, b_u  & B & A_u \\
\end{array}%
\right),
$$
where $\epsilon$ is a real parameter in the interval $0 \le \epsilon \le 1$ and $a_u$, $b_u$
are parameters with mass units. $M'_u$ matrix can be written in the form
$$
 M'_u= M_u(D_u, F_u) + \epsilon \, N_u,
$$
where $M_u(D_u, F_u)$ matrix is given in (\ref{MMM}) and $N_u$ matrix is given by
$$
N_u=
\left(%
\begin{array}{ccc}
   0 & a_u &  b_u \\
   a_u & 0 & 0 \\
   b_u  & 0 & 0 \\
\end{array}%
\right).
$$
Following the analysis given in the appendix and applying right perturbation at
first order in $\epsilon$, we find
$$
O_u = U_u (1 + \epsilon X_u),
$$
where the $U_u$ matrix is given in (\ref{UU_UD}) and $X_u$ matrix is:
$$
 X_u =
\left(%
\begin{array}{ccc}
   0 & x_{1u} &  x_{2u} \\
   -x_{1u} & 0 & x_{3u} \\
   -x_{2u}  & -x_{3u} & 0 \\
\end{array}%
\right),
$$
and its elements are: $x_{1u} = \dfrac{a_u \, \cos \beta }{m_c -m_u} - \dfrac{b_u \, \sin \beta }{m_c -m_u} $,
 $x_{2u} = \dfrac{a_u \, \sin \beta }{m_t -m_u} + \dfrac{b_u \, \cos \beta }{m_t -m_u} $,
and $x_{3u} = 0$.

The new $V'_{ckm}$ matrix takes the following form:
\begin{eqnarray}
 V'_{ckm} &=& O_u^T \, U_d, \\
         &=& (1 - \epsilon X_u)U_u^T \, U_d, \\
         &=&  (1 - \epsilon X_u) V_{ckm}.
\end{eqnarray}
After some algebra, using (\ref{sin_beta}) and considering  $m_t > m_c >> m_u$,
we get that, the element $\left(  V'_{ckm} \right)_{13} $ has the form:
\begin{equation}
\label{u_perturbation}
\left(  V'_{ckm} \right)_{13} =  \left( \frac {V_{cb}^2 }{m_c} - \frac{1}{m_t} \right) \, \epsilon b_u
-  \left( \frac {V_{cb} }{m_c} \right) \, \epsilon a_u.
\end{equation}
We have non zero element, which its magnitude depend on $V_{cb}$ and the
perturbation parameters. The numerical contribution from $b_u $ goes like $10^{-6}$,
while the numerical contribution from $a_u $ goes like $10^{-5}$.  The smallest numerical
element of $M_u(D_u, F_u) $ matrix  is $m_u$, then we consider that the maximum value
of the perturbation is  $m_u/10$. We scanned all
allowed range of $\epsilon \, a_u$ and $\epsilon \, b_u$ parameters and we get that the best numerical absolute value is $8\times 10^{-6}$. For left and left-right
 perturbations (See Appendix), the numerical values were the same order.
 The absolute values of new $V'_{ckm}$ elements are:
$$
|V'_{ckm}|=
\left(%
\begin{array}{ccc}
   0.9753 & 0.2208 & 8\times 10^{-6} \\
   0.2206 & 0.9745 & 0.039 \\
   0.0086  & 0.0380 & 0.9992 \\
\end{array}%
\right).
$$
The values of the mass matrix parameters of $M'_u$ were:
$|\epsilon \, a_u| = |\epsilon \, b_u|  = 0.2 \sim \frac{|m_u|}{10} \ll ,m_u = 2.3, \, |A_u| = 172739, \, |C_u|= 1531.2, \, |D|=6697.47$, all quantities in MeV. The numerical values that corresponding to second order in $\epsilon$ are $O(10^{-8})$ or less.

Now we consider that $M'_d$ mass matrix differs a small quantity
\footnote{$|\epsilon \, a_d| \sim |\epsilon \, b_d|\ll |E_d|, \, |C_d|, \, |D|, \, m_b $}
$\epsilon \, a_d$,  $\epsilon \, b_d$ from $M_d$ in the positions $(1,3), \ (3,1), \ (3,2)$ and $(2,3)$.
$$
 M'_d=
\left(%
\begin{array}{ccc}
   E_d & D_0 & \epsilon \, a_d \\
   D_0 & C_d & \epsilon \, b_d  \\
   \epsilon \, a_d  & \epsilon \, b_d  & m_b \\
\end{array}%
\right),
$$
$M'_d$ matrix can be written in the form
$$
 M'_d= M_d(B_d, F_d) + \epsilon \, N_d,
$$
where $M_d(B_d, F_d)$ matrix is given in (\ref{MMM}) and $N_d$ matrix is given by
$$
N_d=
\left(%
\begin{array}{ccc}
   0 & 0 &  a_d \\
   0 & 0 & b_d \\
   a_d  & b_d & 0 \\
\end{array}%
\right).
$$
 Applying right perturbation at first order in $\epsilon$, we find
$$
O_d = U_d (1 + \epsilon X_d),
$$
where the matrix $X_d$ is:
$$
 X_d =
\left(%
\begin{array}{ccc}
   0 & x_{1d} &  x_{2d} \\
   -x_{1d} & 0 & x_{3d} \\
   -x_{2d}  & -x_{3d} & 0 \\
\end{array}%
\right),
$$
and its elements are $x_{1d} = 0$, $x_{2d} = \dfrac{a_d \, \cos \alpha }{m_b -m_d} -\dfrac{b_d \, \sin \alpha }{m_b -m_d}$ and
$x_{3d} = \dfrac{a_d \, \sin \alpha}{m_b -m_s} +\dfrac{b_d \, \cos \alpha}{m_b -m_s}$.

The new $V'_{ckm}$ matrix takes the following form:
\begin{eqnarray}
 V'_{ckm} &=& V_u^T \, O_d, \\
         &=& V_u^T \, V_d (1 + \epsilon X_d), \\
         &=&  V_{ckm} (1 + \epsilon X_d).
\end{eqnarray}
After some algebra, using (\ref{sin_alpha}) and considering  $m_t >> m_c >> m_u$,
we get that the element $\left(  V'_{ckm} \right)_{13} $ has the form:
\begin{equation}
\label{d_perturbation}
\left(  V'_{ckm} \right)_{13} =  V_{us}  \, \left( \frac {m_s}{m^2_b}  \right) \, \epsilon b_d
+  \left( \frac {1 }{m_b} \right) \, \epsilon a_d.
\end{equation}
We have non zero element, which its magnitude depend on $V_{us}$ and the
perturbation parameters. The numerical contribution from $b_d $ goes like $10^{-6}$,
while the numerical contribution from $a_d $ goes like $10^{-4}$.  The smallest numerical
element from matrix $M_d(B_d, F_d)$ is $E_d$, then we consider that the maximum value
of the perturbation is  $E_d/10$. We scanned all
allowed range of $\epsilon \, a_d$ and $\epsilon \, b_d$ parameters and we get that,
 the best numerical absolute value is
 $2\times 10^{-4}$. For left and left-right
 perturbations, the numerical values were the same order of magnitude.
 The absolute values of new $V'_{ckm}$ elements are:
$$
|V_{ckm}|=
\left(%
\begin{array}{ccc}
   0.9742 & 0.2253 & 0.0002 \\
   0.2251 & 0.9734 & 0.0406 \\
   0.0089  & 0.0396 & 0.9991 \\
\end{array}%
\right).
$$
The values of the parameters were:
$|\epsilon \, a_d| = |\epsilon \, b_d|= 0.9 \sim \frac{|E_d|}{10} \ll |E_d| = 9.37, \, |C_d|= 90.42, \, |D|=20.32, \, m_b = 4180$, all quantities in MeV. The numerical values that corresponding to second order in $\epsilon$ are $O(10^{-8})$ or less.

Also we have numerically analyzed all possibilities to get a perturbation
on both mass matrices  without get better numerically values in
the $V_{ckm}$ matrix.

From the analysis of this section, we conclude that 4-zero texture set in the normal and perturbative
cases are ruled out, because they can not reproduce the experimental values of the $V_{ckm}$ matrix.

\subsection{$V_{ckm}$ from 3-zero texture set}

The next case of structure is a 3-zero texture set, which it born when one
type of quarks has as mass matrix $M(B,F)$, $M(B,D)$ or $M(D,F)$ and the
other type of quarks has mass matrix $M(B)$, $M(D)$ or $M(F)$. We have in total
18 possible combinations\footnote{We shall discuss these kind of
textures in a forthcoming paper \cite{progress}}.

From the analysis presented before, we can point out two issues:
\begin{itemize}
 \item We can introduce a $(V_{ckm})_{13}$ element different from zero, setting
in appropriate way the values (1,3) and (3,1) in $M_d$ matrix. From (\ref{d_perturbation}),
we can note a lineal dependence between $(V_{ckm})_{13}$  and the perturbation,
$\epsilon a_d$ and if $|\epsilon a_d| \sim E_d$ we obtain the numerical value of
$(V_{ckm})_{13}$ very close that the experimental one. Then we will consider that
$F_d$ is the same order than $E_d$.

 \item The quark mass matrix can be split in two parts, a diagonal part
plus off-diagonal contributions, which both of them are in power
series of $V_{us}$ and $V_{cb}$ elements.
\end{itemize}

Considering the above statements, we take $M(D,F)$ as 2-zero structure for
u-type quarks, {\it i.e} it has the form given in (\ref{MU_parts}) and the
matrix that diagonalize it is (\ref{UU_UD}). For d-quarks we take $M(B)$ as
1-zero structure given by
$$
M_d=
\left(%
\begin{array}{ccc}
   E_d  & D_d & F_d \\
   D_d & C_d  & 0 \\
   F_d  & 0 & A_d  \\
\end{array}%
\right),
$$
where each element is parameterized as:
$$
\begin{array}{l}
\text{\textit{main diagonal elements}} \\
 A_d = m_b + x \, V_{us}^3, \\
 C_d = m_s - 2 \delta_{sd} V_{us}^2 + y \, V_{us}^3, \\
 E_d = m_d + 2 \delta_{sd} V_{us}^2 + z \, V_{us}^3
 \end{array} \qquad
 \begin{array}{l}
  \text{\textit{off diagonal elements}} \\
 D_d = + 2 \delta_{sd} V_{us}, \\
 F_d = m_d + 2 \delta_{sd} V_{us}^2,
\end{array}
$$
where $(x,\, y,\, z)$ are variables to find. Now as $M_d$ is congruent
with $Diag \left[m_d,\, m_s,\, m_b \right]$ we can write the following equations:
\begin{eqnarray}
Tr\,M_d &=& m_d + m_s + m_b, \nonumber \\
det\,M_d &=& m_d \, m_s \, m_b, \\
\dfrac{1}{2}\left[ \tr^2 M_d-\tr M_d^2 \right] &=&  m_d \, m_s + m_d \,m_b + m_s \, m_b, \nonumber 
\end{eqnarray}
this set of equations has six solutions for $(x,\, y,\, z)$, and we choose the solution that $ E_d < C_d < A_d$ is hold,
i.e. the numerical values for $(x,\, y,\, z)$ are $(1.84199,\, x + z,\, -1.84417)$, then
numerically the matrix $M_d$ results
$$
M_d=
\left(%
\begin{array}{ccc}
   9.42827 & 19.802896 & 9.380186 \\
   19.802896 & 90.1575 & 0 \\
   9.380186  & 0 & 4180.21 \\
\end{array}%
\right),
$$
and the numerical absolute values of $V_{ckm}$ elements are:
$$
|V_{ckm}|=
\left(%
\begin{array}{ccc}
   0.974118 & 0.226027 & 0.00224906 \\
   0.225925 & 0.973273 & 0.0412108 \\
   0.00712578  & 0.0406523 & 0.999148 \\
\end{array}%
\right),
$$
that is in agreement with the experimental value of $V_{ckm}$ matrix.

This is a good example that shows that 3-zero texture sets are viable candidates
to model the quark mass matrices.

\section{Conclusions}
In this paper, in the understanding that by definition the physical
mass of the quarks and charged leptons
are the eigenvalues of the mass matrices. We found the necessary and
sufficient conditions over the
characteristic polynomial coefficients from any symmetric $3$ by $3$
matrix, so that it has real, simple
and positive roots. We apply this formalism to analyze the symmetric
texture-zero quark matrices, we found
that a lot of them are ruled out ({\it i.e.} they have two equal
eigenvalues or one of them negative). Only
the zero texture
matrices $M(B,F)$, $M(B,D)$, $M(D,F)$,
$M(B)$, $M(D)$ and $M(F)$ are in agreement with this condition. In the
texture-zero formalism, the matrices
have variable coefficients, the conditions \ref{condition2} and
\ref{condition3} impose restrictions
over these coefficients, this means that, we need to find the
complete domain of the coefficients in
both mass matrices, u-type quarks and d-type quarks, in order
to approximate the experimental values of the $V_{ckm}$ matrix. We develop analytically
the case of four zero sets, and we show the set that gives the best approximation to the $V_{ckm}$ matrix and
always a zero element in the theoretical $V_{ckm}$ matrix is found, to remove this zero,
we implement a perturbation method and we analyze the 4 zero texture set, even with these results the four-zero texture sets are ruled out.
The quark mass matrix can be split in in two parts, a diagonal part
plus off-diagonal contributions, which both of them are in power
series of $V_{ckm}$ elements, statement that is valid for three and four zero
texture sets.
With an example, we show that 3-zero texture sets are viable to model the quark mass matrices, and this
structure in the minimal which satisfy that they have real, positive and different eigenvalues and also
it reproduce the experimental values of the $V_{ckm}$ matrix.
\begin{acknowledgments}
The authors acknowledge to Baltazar Aguirre Hern\'andez and Lorenzo D\'{\i}az-Cruz
for their useful comments. The authors also acknowledge to CONACyT(\textit{M\'exico}),
PROMEP(\textit{M\'exico}) and SNI(\textit{M\'exico}) for their financial support.
\end{acknowledgments}

\appendix
\section{Proof Theorem}

In this appendix we proof the theorem \ref{theorem_1}, for this we need the following statement

\begin{lemma}\label{lemma}
The polynomial of second degree $p(\lambda)=\lambda^2+a_1\lambda +a_0$ has two
real, simple and positive roots if and only if the following condition hold.
\begin{equation}
a_1<0<a_0<\dfrac{a_1^2}{4}.\label{pol2degree}
\end{equation}
\end{lemma}

\begin{proof} Follows from a simple computation.
\end{proof}

Now we remember and proof the theorem \ref{theorem_1}.

\noindent \textbf{Theorem 1}
The polynomial of degree 3, $p(\lambda)=\lambda^3+a_2\lambda^2+a_1\lambda +a_0$ has three different, real and positive roots if and only if the following conditions over its coefficients $a_0$,  $a_1$ $a_2$ hold.
\begin{enumerate}
\item $a_0,a_2<0<a_1.$
\item $3a_1<a_2^2$.
\item If $\lambda_4=\dfrac{-a_2+\sqrt{a_2^2-3a_1}}{3}$ and
$\lambda_5=\dfrac{-a_2-\sqrt{a_2^2-3a_1}}{3}$, then $p(\lambda_4)<0$ and $p(\lambda_5)>0$.
\end{enumerate}

\begin{proof}
If exists three different $\lambda_i\in\R^+$,  $(i=1,2,3)$ such that
$p(\lambda)=(\lambda-\lambda_1)(\lambda-\lambda_2)(\lambda-\lambda_3)=\lambda^3-(\lambda_1+\lambda_2+\lambda_3)\lambda^2
+(\lambda_1\lambda_2+\lambda_1\lambda_3+\lambda_2\lambda_3)\lambda
-\lambda_1\lambda_2\lambda_3$, then by equality of polynomials we obtain

\begin{itemize}
\item$a_2=-(\lambda_1+\lambda_2+\lambda_3)<0$,
\item$a_0=-\lambda_1\lambda_2\lambda_3<0$,
\item$a_1=\lambda_1\lambda_2+\lambda_1\lambda_3+\lambda_2\lambda_3>0$,
\end{itemize}
and the condition \emph{1} hold.

Without loss of generality, we suppose that  $0<\lambda_1<\lambda_2<\lambda_3$\footnote{In this work, a chain of inequalities $a < b < c < \dots$, the first inequality
is $a<b$, the second one is $b<c$ and so on.}, if we have a polynomial of degree 3 with three real, simple roots then there exits two critical points, they are roots of the first derivative, i.e. the condition \emph{2} hold. Now if the roots of polynomial are positive then the critical points are positive too and the follow chain of inequalities hold  $\lambda_1<\lambda_5<\lambda_2<\lambda_4<\lambda_3$. From the coefficient of $\lambda^3$ is 1, we have that $\lim_{\lambda\to\infty}p(\lambda)=\infty$ and
$\lim_{\lambda\to-\infty}p(\lambda)=-\infty$, then for points less than $\lambda_1$ the polynomial is negative, we applied the Rolle theorem to the roots $\lambda_1$ and $\lambda_2$, therefore the polynomial has a maximum value between $\lambda_1$ and $\lambda_2$, and therefore $p(\lambda_5)>0$. Similarly for points greater than $\lambda_3$ the polynomial is positive and we applied the Rolle theorem to the roots $\lambda_2$ and $\lambda_3$, and the polynomial has a minimum, this value is negative i.e. $p(\lambda_4)<0$.

Conversely, we have that $p(\lambda)=\lambda^3+a_2\lambda^2+a_1\lambda +a_0$ such
that the conditions \textit{1, 2} and \textit{3} hold. The conditions \emph{1,2} join to lemma \ref{lemma} implies that $p'(\lambda)$ has two real, simple and positive roots given by :

\begin{equation*}
\lambda_4=\dfrac{-a_2+\sqrt{a_2^2-3a_1}}{3}, \quad \lambda_5=\dfrac{-a_2-\sqrt{a_2^2-3a_1}}{3}.
\end{equation*}
First we observe that $0<\lambda_5<\lambda_4$, now computing
$p''(\lambda_4)=2\sqrt{a_2^2-3a_1}>0$, this implies in $\lambda_4$ we have a
minimum whereas $p''(\lambda_5)=-2\sqrt{a_2^2-3a_1}<0$ and then in $\lambda_5$ we
have a maximum. We applied repeatedly the Intermediate Value theorem. From \textit{1}, we have that $p(0)=a_0<0$ and from \textit{3} it follows that $p(\lambda_5)>0$, then we have a positive root. The condition \textit{3}, $p(\lambda_4)<0$ and $p(\lambda_5)>0$, guarantee that exits a second root between $\lambda_5$ and $\lambda_4$, finally due to the coefficient to $\lambda^3$ is positive $p(\lambda)$ we have that $\lim_{\lambda\to\infty}p(\lambda)=\infty$
and this implies $p(\lambda_4)<0$, then $p(\lambda)$ intersects to horizontal axis  one more time in the third root.
\end{proof}

\section{2-zero textures}

In the above sections we show that the matrices $M(B,F), \ M(B,D)$ and $M(D,F)$
can have three positive, real and different eigenvalues, these matrices are diagonal
by blocks ( one block $1 \times 1$ and other
block $2 \times 2$). They can be diagonalize by matrices that are also
diagonal by blocks.

Pay attention only in the $2 \times 2$ block. The mass matrix can be rewritten as:
$$
M_{2 \times 2} =
\begin{pmatrix}
   y & K  \\
   K & x  \\
\end{pmatrix},
\qquad
(K = B, D, F).
$$
This matrix has to be congruent with
$$
\bar{M}_{2 \times 2} =
\begin{pmatrix}
  m_i & 0 \\
  0 & m_j \\
\end{pmatrix}, \qquad (i,j)= (1,2),(2,3),(1,3).
$$

This implies the following relations among their elements:
\begin{eqnarray}
  x+y &=& m_i+m_j, \\
  x y-K^2 &=& m_i \, m_j,
\end{eqnarray}
the solutions for $x$ and $y$ are:
\begin{eqnarray}
\label{x}
x(K) &=& \mu_{ij} \pm \sqrt{\delta_{ij}^2 -K^2}, \\
\label{y}
y(K) &=& \mu_{ij} \mp \sqrt{\delta_{ij}^2 -K^2},
\end{eqnarray}
where $\mu_{ij} = \dfrac{m_i+m_j}{2}$, if $m_i>m_j, \ \delta_{ij} = \dfrac{m_i-m_j}{2}$ and the
parameter $K$ has to satisfy $|K| \leq \delta_{ij}$.

The matrix that diagonalize the matrix $M_{2 \times 2}$ always can be set as:
$$
\left(%
\begin{array}{cc}
   \cos\theta & \sin\theta \\
   -\sin\theta & \cos\theta \\
\end{array}%
\right),$$
where:
$\sin\theta=\dfrac{K}{\sqrt{K^2+(y-m_i)^2}}$ and $\theta \in[0,\pi/4]$.

\section{Perturbation Theory}
In this appendix we applied perturbation theory to texture formalism
\footnote{The first work in this direction was \cite{Fritzsch:2011cu}
and it applies non-hermitian perturbations to the 6-zero texture}.

We start dividing the complete mass matrix in two parts
\begin{equation}
 M = M_0 + \epsilon N,
\end{equation}
where $M_0$ and $N$ are known mass matrices and $\epsilon$ is a small parameter.
We look for a $O$ matrix that diagonalize the $M$ matrix in the following way:
\begin{equation}
\label{A1}
O^T M O = \bar{M},
\end{equation}
where $\bar{M}$ is a diagonal matrix.

We have three different versions of the perturbation method according to
the way that $O$ matrix is proposed, namely:
\begin{enumerate}
 \item \textit{Right Perturbation}, when the $O$ matrix takes the form
\begin{equation}
\label{RPT}
O = O_0 (1 + \epsilon X).
\end{equation}

\item \textit{Left Perturbation}, when the matrix $O$ takes the form
\begin{equation}
\label{LPT}
O =  (1 + \epsilon X) \, O_0.
\end{equation}

\item  \textit{Left-Right Perturbation}, when the matrix $O$ takes the form
\begin{equation}
\label{LRPT}
O =  (1 + \epsilon X) \, O_0 \, (1 + \epsilon X).
\end{equation}
\end{enumerate}

Where the $O_0$ matrix diagonalize the $M_0$ matrix ($ O_0^T \, M_0 \,O_0 = \bar{M}$)
and the $X$ matrix is determined in this process.

From the orthogonality condition of the $O$ matrix it is find that $X$ is antisymmetric matrix
$X^T = -X$ and $Y +Y^T = X^2$ for all cases.

Notation: We are considering $\bar{A} = O^T_0 \, A \, O_0$ for any matrix $A$.

\subsection{Right Perturbation}

Substituting the form the $O$ matrix (Eq. \ref{RPT}) into (Eq.\ref{A1}):
$$
\left[ O_0 (1 + \epsilon X) \right]^T \, M \, \left[ O_0 (1 + \epsilon X) \right] = \bar{M}.
$$
After some algebra one gets, at first order in $\epsilon$ parameter, that the
$X$ matrix has to satisfied:
\begin{equation}
 \bar{N} = [X,\bar{M}].
\end{equation}
At second order in $\epsilon$, the
$Y$ matrix has to satisfied:
\begin{equation}
 \bar{N} \,X  + X \, \bar{N} = [Y + Y^T,\bar{M}].
\end{equation}

\subsection{Left Perturbation}

Substituting the form the $O$ matrix (Eq. \ref{LPT}) into (Eq.\ref{A1}):
$$
\left[  (1 + \epsilon X) O_0\right]^T \, M \, \left[  (1 + \epsilon X) O_0 \right] = \bar{M}.
$$
After some algebra one gets, at first order in $\epsilon$ parameter, that the
$\bar{X}$ matrix has to satisfied:
\begin{equation}
 \bar{N} = [\bar{X},\bar{M}].
\end{equation}
At second order in $\epsilon$, the
$\bar{Y}$ matrix has to satisfied:
\begin{equation}
 \bar{N} \,\bar{X}  + \bar{X} \, \bar{N} = [\bar{Y} + \bar{Y}^T,\bar{M}].
\end{equation}

\subsection{Left-Right Perturbation}

Substituting the form the $O$ matrix (Eq. \ref{LRPT}) into (Eq.\ref{A1}):
$$
\left[  (1 + \epsilon X) O_0 (1 + \epsilon X)  \right]^T \, M \,
\left[  (1 + \epsilon X) O_0 (1 + \epsilon X)  \right] = \bar{M}.
$$
After some algebra one gets, at first order in $\epsilon$ parameter, that the
$\bar{X}$ matrix has to satisfied:
\begin{equation}
 \bar{N} = [X + \bar{X},\bar{M}].
\end{equation}



\begin{thebibliography}{99}

\bibitem{Fritzsch:1977za}
  H.~Fritzsch,
  Phys.\ Lett.\  B {\bf 70}, 436 (1977).
  H.~Fritzsch,
  Phys.\ Lett.\  B {\bf 73}, 317 (1978).
\bibitem{Fritzsch:1999ee}
  H.~Fritzsch and Z.~z.~Xing,
  Prog.\ Part.\ Nucl.\ Phys.\  {\bf 45}, 1 (2000)
  [arXiv:hep-ph/9912358] and
  M.~Gupta and G.~Ahuja,
  Int.\ J.\ Mod.\ Phys.\ A {\bf 26}, 2973 (2011)
  [arXiv:1206.3844 [hep-ph]].
\bibitem{Haussling:1997ue}
  R.~Haussling and F.~Scheck,
  Phys.\ Rev.\ D {\bf 57}, 6656 (1998)
  [hep-ph/9708247].
  T.~-K.~Kuo, S.~W.~Mansour and G.~-H.~Wu,
  Phys.\ Rev.\ D {\bf 60}, 093004 (1999)
  [hep-ph/9907314].
  T.~-K.~Kuo, S.~W.~Mansour and G.~-H.~Wu,
  Phys.\ Lett.\ B {\bf 467}, 116 (1999)
  [hep-ph/9907521].
\bibitem{FriedbergLA}
S.~Friedberg, A.~Insel and L.~Spence. Linear Algebra. Fourth Edition.
Prentice Hall (2006)
\bibitem{Branco:1999nb}
  G.~C.~Branco, D.~Emmanuel-Costa and R.~Gonzalez Felipe,
  Phys.\ Lett.\ B {\bf 477}, 147 (2000)
  [hep-ph/9911418].
\bibitem{Fritzsch:2002ga}
  H.~Fritzsch and Z.~-z.~Xing,
  Phys.\ Lett.\ B {\bf 555}, 63 (2003)
  [hep-ph/0212195].
\bibitem{Fritzsch:2011cu}
  H.~Fritzsch, Z.~-z.~Xing and Y.~-L.~Zhou,
  Phys.\ Lett.\ B {\bf 697}, 357 (2011)
  [arXiv:1101.4272 [hep-ph]].
\bibitem{PDG}
  J.~Beringer {\it et al.}  [Particle Data Group],
  Phys.\ Rev.\ D {\bf 86}, 010001 (2012).\texttt{ http://pdg.lbl.gov/}
\bibitem{progress}
 A. Criollo {\it et al.} Work in progres.

\end{thebibliography}
\end{document}